\documentclass[11pt]{article}
\usepackage[OT1]{fontenc}
\usepackage[utf8]{inputenc}

\usepackage[margin=1in]{geometry}
\usepackage{amsfonts,amsmath,amssymb,amsthm}
\usepackage{cite,microtype,graphicx,hyperref}
\usepackage[capitalize,nameinlink]{cleveref}
\usepackage{algorithm, enumitem}
\usepackage[noend]{algpseudocode}

\graphicspath{{figs/}}

\newcommand{\restate}[2]{\newtheorem*{restate-#1}{\autoref{#1}}\begin{restate-#1}#2\end{restate-#1}}

\newtheorem{theorem}{Theorem}
\newtheorem{lemma}{Lemma}
\newtheorem{corollary}{Corollary}

\theoremstyle{definition}
\newtheorem{definition}{Definition}
\newtheorem*{problem*}{Problem}

\newcommand{\ind}{\operatorname{index}}
\newcommand{\lvl}{\operatorname{size}}

\newcommand{\lev}{l}
\newcommand{\cen}{c}
\newcommand{\prn}{p}
\newcommand{\T}{\mathcal{T}}
\newcommand{\C}{\mathcal{C}}
\renewcommand{\L}{\mathcal{L}}
\newcommand{\calS}{\mathcal{S}}
\newcommand{\calP}{\mathcal{P}}

\title{Maintaining Light Spanners via Minimal Updates}

\author{David Eppstein and Hadi Khodabandeh\thanks{Department of Computer Science, University of California, Irvine. Research supported in part by NSF grant CCF-2212129.}}

\date{ }

\begin{document}

\maketitle



\begin{abstract}
We study the problem of maintaining a lightweight bounded-degree $(1+\varepsilon)$-spanner of a dynamic point set in a $d$-dimensional Euclidean space, where $\varepsilon>0$ and $d$ are arbitrary constants. In our fully-dynamic setting, points are allowed to be inserted as well as deleted, and our objective is to maintain a  $(1+\varepsilon)$-spanner that has constant bounds on its maximum degree and its lightness (the ratio of its weight to that of the minimum spanning tree), while minimizing the recourse, which is the number of edges added or removed by each point insertion or deletion. We present a fully-dynamic algorithm that handles point insertion with amortized constant recourse and point deletion with amortized $\mathcal{O}(\log\Delta)$ recourse, where $\Delta$ is the aspect ratio of the point set.
\end{abstract}


\section{Introduction}%
\label{sec:intro}%

Spanners are sparse subgraphs of a denser graph that approximate its shortest path distances. Extensive study has been made of \emph{geometric spanners}, for which the dense graph is a complete weighted graph on a point set in $d$-dimensional Euclidean space, and where the weight of an edge $(u,v)$ is simply the Euclidean distance between $u$ and $v$. The approximation quality of a spanner is measured by its stretch factor $t$, where a $t$-spanner $S$ is defined by the property that for every two vertices $u$ and $v$ in the graph, $d_S(u,v)\leq t\cdot d(u,v)$. Here $d$ and $d_S$ are respectively the Euclidean metric of dimension $d$ and the shortest path metric induced by the spanner. In other words, the Euclidean distances are \emph{stretched} by a factor of at most $t$ in the spanner.

In this paper, we study the problem of maintaining $1+\varepsilon$-spanners under a dynamic model in which points are inserted and removed by an adversary and our goal is to minimize the recourse, which is the number of changes we make to the edge set of the spanner. The recourse should be distinguished from the time it takes us to calculate the changes we make, which might be larger; our use of recourse instead of update time is motivated by real-world networks, where making a physical change to the network is often more costly than the actual run-time of the algorithm that decides what changes need to be made.

We introduce a hierarchical structure that we update with minimal changes after each operation. We use this hierarchy as the basis of our sparse spanner. It is worth noting that using hierarchical structures to build sparse spanners was known in prior work, but our hierarchy is designed in a way that it suits our needs in this paper. Then we turn our attention into light-weight spanners and we use novel concepts and ideas (such as the notion of stretch factor for subsets of edges and quantifying the impact of an edge update on other edges through a potential function) to iteratively lighten the weight of the spanner after point insertion and deletion. We also use well-known techniques in prior work such as bucketing and amortized analysis, which eventually lead us to our results on the amortized number of edge updates in each bucket. This was made possible through carefully crafting a potential function that decreases via our maintenance updates on the spanner.

\subsection{Related work}
Geometric $t$-spanners have numerous applications in network design problems \cite{narasimhan2007geometric}. Finding a sparse lightweight $t$-spanner is the core of many of these applications. The existence of such spanners and efficient algorithms for constructing them have been considered under different settings and constraints \cite{biswas2021massively,salzman2014sparsification,har2005fast}. In offline settings, where the point set is given as a whole to the algorithm, the prominent \emph{greedy spanner} algorithm is well known for its all-in-one quality due to its optimal performance under multiple measures including sparsity (its number of edges), lightness (the weight of the spanner divided by the weight of the minimum spanning tree), and maximum degree \cite{althofer1993sparse,borradaile2019greedy}. The output of the greedy spanner also has low crossing number in the plane and small separators and separator hierarchies in doubling metric spaces \cite{eppstein2021edge,le2022greedy}.
However, in some applications, the points of an input set may repeatedly change as a spanner for them is used, and a static network would not accurately represent their distances. The dynamic model, detailed below, deal with these types of problems.

In the dynamic model, points are inserted or removed one at a time, and the algorithm has to maintain a $t$-spanner at all times. In this setting the algorithm is allowed to remove previous edges. For $n$ points in $d$-dimensional Euclidean space, Arya, Mount, and Smid~\cite{arya1994randomized} designed a spanner construction with a linear number of edges and $\mathcal{O}(\log n)$ diameter under the assumption that a point to be deleted is chosen randomly from the point set, and a point to be inserted is chosen randomly from the new point set. Bose, Gudmundsson, and Morin~\cite{bose2004ordered} presented a semi-dynamic $(1+\varepsilon)$-spanner construction with $\mathcal{O}(\log n)$ maximum degree and diameter. Gao, Guibas, and Nguyen \cite{gao2006deformable} designed the deformable spanner, a fully-dynamic construction with $\mathcal{O}(\log\Delta)$ maximum degree and $\mathcal{O}(\log\Delta)$ lightness, where $\Delta$ is the aspect ratio of the point set, defined as the ratio of the length of the largest edge divided by the length of the shortest edge.

In the spaces of bounded doubling dimension, Roditty~\cite{roditty2012fully} provided the first dynamic spanner construction whose update time (and therefore recourse) depended solely on the number of points ($\mathcal{O}(\log n)$ for point insertion and $\Tilde{\mathcal{O}}(n^{1/3})$ for point removal). This was later improved by Gottlieb and Roditty~\cite{gottlieb2008improved}, who extended this result in doubling metrics and provided a better update time as well as the bounded-degree property. The same authors further improved this construction to have an asymptotically optimal insertion time (and therefore recourse) of $\mathcal{O}(\log n)$ under the algebraic decision tree model~\cite{gottlieb2008optimal} but logarithmic lightness.

It is worth to mention that none of the work mentioned above in the dynamic setting achieve a sub-logarithmic lightness bound on their output. The problem of maintaining a light spanner in this setting has remained open until now.

\subsection{Contributions}
Light-weight fully-dynamic spanners have not been studied in the literature to the best of our knowledge. There are currently no known algorithms that provide a spanner with constant lightness except by rebuilding the whole spanner. We construct a fully-dynamic spanner that aims to minimize the recourse, defined by the number of edges updated after a point insertion or removal. Our spanner maintains, at all times, a lightness and a maximum degree that are bounded by constants. Our maintenance regime achieves amortized constant recourse per point insertion, and amortized $\mathcal{O}(\log\Delta)$ recourse per point deletion. We state and prove our bounds in \cref{thm:main}.

\restate{thm:main}{Our fully-dynamic spanner construction in $d$-dimensional Euclidean spaces has a stretch-factor of $1+\varepsilon$ and a lightness that is bounded by a constant. Furthermore, this construction performs an amortized $\mathcal{O}(1)$ edge updates following a point insertion, and an amortized $\mathcal{O}(\log\Delta)$ edge updates following a point deletion.}%

The hidden constants in our bounds only depend on $\varepsilon$ and $d$. Our amortized bound for recourse after point insertion is optimal but for point deletion we do not claim optimality. However, it is worth to mention that our recourse bound for deletion is not worse than the bounds achieved in prior work.
In order to reach our bounds on recourse we introduce new techniques for iteratively improving the weight of the spanner without losing its other characteristics.


\section{Preliminaries and overview}%
\label{sec:prelim}%
In this section, we cover the notations as well as important definitions and facts that we use throughout the paper. We also provide an overview of what to expect in the upcoming sections and the methods we use to reach our bounds on the recourse.

\textbf{Notation}. We denote the current point set by $V$ and its aspect ratio (as defined earlier) by $\Delta$. We use the notations $\lVert e\rVert$ and $\lVert P\rVert$ for the Euclidean length of an edge $e$ and a path $P$, respectively. We also refer to the Euclidean distance of two points $u$ and $v$ by $\lVert uv\rVert$ or $d(u,v)$, interchangeably. The notation $|E|$ is used when we are referring to the size of a set $E$. Also, for a spanner $S$, the weight of $S$ is shown by $w(S)$.

\subsection{Overview}
We build our spanners on top of a hierarchical clustering $(\T, R)$ of the point set that we maintain dynamically as the point set changes over time. The tree $\T$ represents the parent-child relationship between the clusters, and the constant $R$ specifies how cluster radii magnify on higher levels. Each cluster $\C\in\T$ is specified by a pair $\C=(p, l)$ where $p\in\mathbb{R}^d$ is one of the given points at the center of the cluster and $l\in\mathbb{Z}$ is the level of the cluster. The level of a cluster determines  its radius, $R^l$. It is possible for the same point to be the center of multiple clusters, at different levels of the hierarchy.

We maintain our hierarchy so that after a point insertion, a cluster is added centered at the new point, and after a point deletion, each cluster with the deleted point as its center is removed. Meanwhile, we maintain a \emph{separation} property on the hierarchy to help us build a sparse spanner. Additional edges of our sparse spanner connect pairs of clusters of the same level. Each such edge ensures that pairs of descendants of its endpoints have the desired stretch-factor. These edges form a bounded-degree graph on the clusters at each level, but this property alone would not ensure bounded degree for our whole spanner, because of points that center multiple clusters. Instead, we redistribute the edges of large degree points to derive a bounded-degree spanner.

Maintaining bounded lightness on the other hand is done through an iterative pruning process. We start by removing certain edges to decrease the weight of the spanner, which in turn might cause some other pairs that previously used the removed edge in their shortest paths to not meet the stretch bound of $1+\varepsilon$. We fix those pairs by adding an edge between them, which again increases the weight of the spanner. This causes a chain of updates that alternatively improve the stretch and worsen the weight of the spanner, or improve the weight and worsen the stretch of the spanner. We show that this sequence of updates, which we call \emph{maintenance} updates, if performed properly and for the right pairs, will indeed not end in a loop, and even more strongly, will terminate after an amortized constant number of iterations. This will be covered in  \cref{sec:light}.

The rest of this section includes the techniques we use for our light-weight spanner construction. We start with one of these techniques which is called the \emph{bucketing} technique. Instead of enforcing the stretch bound and the lightness bound on the whole spanner, we partition its edges into a constant number of subsets and we enforce our criteria on these subsets. This partitioning is necessary for the purpose of our analysis.

\textbf{Bucketing.} We maintain a partition of the spanner edges into a constant number of subsets. As we mentioned before, our invariants are enforced on these subsets instead of the whole spanner. Let $C\gg c>1$ be constants that we specify later. We partition the edges of the spanner into $k=\lceil \log_c C\rceil$ subsets, $S_0,S_1,\cdots,S_{k-1}$, so that for each set $S_i$ and any pair of edges $e,f\in S_i$ such that $\lVert e\rVert\geq \lVert f\rVert$, one of the following two cases happen: (i) either $\lVert e\rVert/\lVert f\rVert< c$ or (ii) $\lVert e\rVert/\lVert f\rVert\geq C$. In other words, the edge lengths in the same set are either very close, or very far from each other.

Such partitioning can be maintained easily by assigning an edge $e$ to the set with index $\ind(e) = \lfloor \log_c\lVert e\rVert \rfloor \bmod k$. We refer to this as the \emph{index} of the edge $e$. We also define the \emph{size} of an edge $e$ as $\lvl(e) =\lfloor (\log_c\lVert e\rVert )/k\rfloor$. By definition, if $\ind(e)=i$ and $\lvl(e)=j$, then $c^{kj+i}\leq \lVert e\rVert <c^{kj+i+1}$. We similarly define the index and the size for any pair $(u,v)$ of vertices that are not necessarily connected in the spanner: $\ind(u,v) = \lfloor \log_c\lVert uv\rVert \rfloor \bmod k$, and $\lvl(u,v) =\lfloor (\log_c\lVert uv\rVert )/k\rfloor$.

\textbf{Invariants.} In order to construct a light-weight spanner, we start from our sparse dynamic spanner construction. To distinguish the edges of our light spanner with the edges of our sparse spanner, we call the edges of our sparse spanner the \emph{potential pairs}, since a carefully filtered set of those edges will make up our light-weight spanner. After bucketing the potential pairs, since we maintain the edges of each bucket separately, we must find per-bucket criteria that guarantee the the main properties we expect from our spanner: the stretch-factor and the lightness. We call these criteria the \emph{invariants}. To make sure the union of the buckets meets the stretch bound, we generalize the notion of stretch factor to work on individual buckets and we call it Invariant 1.

\begin{itemize}
    \item \textbf{Invariant 1.} For each pair of vertices $(u,v)\notin S_i$ with index $i$, there must exist a set of edges $e_1=(x_1,y_1), e_2=(x_2,y_2),\dotsc,e_l=(x_l,y_l)$ in $S_i$ such that
    $$\sum_{i=1}^l \lVert e_i\rVert  + (1+\varepsilon)\left(\lVert ux_1\rVert  + \sum_{i=1}^{l-1}\lVert y_ix_{i+1}\rVert  + \lVert y_lv\rVert \right)< (1+\varepsilon)\lVert uv\rVert .$$
    In other words, $u$ must reach $v$ by a path of cost at most $(1+\varepsilon)\lVert uv\rVert $ where the cost of every edge $e\in S_i$ is $\lVert e\rVert $ and the cost of every edge $e\notin S_i$ is $(1+\varepsilon)\lVert e\rVert $.
\end{itemize}

\begin{lemma}%
\label{lem:inv1}%
If Invariant 1 holds for all $S_i$, then $S=\bigcup_{i=0}^{k-1} S_i$ is a $(1+\varepsilon)$-spanner.
\end{lemma}
\begin{proof}
Let $(u,v)$ be a pair of vertices. We find a $(1+\varepsilon)$-path between $u$ and $v$ using edges in $S$. Let $i=\ind(u,v)$. By Invariant 1 there exists a set of edges $e_1=(x_1,y_1), e_2=(x_2,y_2),\dotsc,e_l=(x_l,y_l)$ in $S_i$ such that
$$\sum_{i=1}^l \lVert e_i\rVert  + (1+\varepsilon)\left(\lVert ux_1\rVert  + \sum_{i=1}^{l-1}\lVert y_ix_{i+1}\rVert  + \lVert y_lv\rVert \right)< (1+\varepsilon)\lVert uv\rVert . $$
Consider the path $P=ux_1y_1x_2y_2\cdots x_ly_lv$ between $u$ and $v$. We call this path the \emph{replacement path} for $(u,v)$. The edges $x_1y_1,x_2y_2,\dotsc,x_ly_l$ are present in $S_i$ (and therefore present in $S$) but the other edges of the replacement path are missing from $S_i$. A similar procedure can be performed on the missing pairs recursively to find and replace them with their corresponding replacement paths. This recursive procedure yields a $(1+\varepsilon)$-path for $(u,v)$ and it terminates because the length of each missing edge in a replacement path is smaller than the length of the edge that is being replaced (otherwise Invariant 1 would not hold).
\end{proof}

Furthermore, we bound the weight of the spanner by ensuring the second invariant, which is the leapfrog property on $S_i$. \cite{das1995new}

\begin{itemize}
    \item \textbf{Invariant 2.} Let $(u,v)\in S_i$. For every subset of edges $e_1=(x_1,y_1), e_2=(x_2,y_2),\dotsc,e_l=(x_l,y_l)$ in $S_i$ the inequality
    $$\sum_{i=1}^l \lVert e_i\rVert + (1+\varepsilon)\left(\lVert ux_1\rVert + \sum_{i=1}^{l-1}\lVert y_ix_{i+1}\rVert + \lVert y_lv\rVert \right)> (1+\varepsilon')\lVert uv\rVert$$
    holds, where $\varepsilon'<\varepsilon$ is a positive constant. In other words, $u$ should not be able to reach $v$ by a (short) path of cost $(1+\varepsilon')\lVert uv\rVert$, where the edge costs are the same as in Invariant 1.
\end{itemize}

The leapfrog property leads to a constant upper bound on the lightness of $S_i$, for each $0\leq i < k$. And since the weight of the minimum spanning tree on the end-points of each $S_i$ is at most a constant factor of the weight of the minimum spanning tree on the whole point set, this implies a constant upper bound on the lightness of the spanner $S=\bigcup_{i=0}^{k-1} S_i$. As well as the weight bound, we prove, in the following lemma, that Invariant 2 implies a similar result to the packing lemma, but for the number of edges on the same level.

\begin{lemma}[Edge packing] \label{lem:edge-packing}
Let $E$ be a set of edges (segments) with the same index and the same level that is consistent with Invariant 2. Also, assume that $E$ is contained in a ball of radius $R$, and the minimum edge size in $E$ is $r$. Then
$$|E|< C_1(R/r)^{2d}$$
where $C_1=(2(1+\varepsilon)/\varepsilon')^{2d}d^d$ is a constant.
\end{lemma}
\begin{proof}
A simple observation is that for any two segments $(u,v)$ and $(y,z)$ in $E$ we must have
$$\max(\lVert uy\rVert ,\lVert vz\rVert ) > \frac{\varepsilon'}{2(1+\varepsilon)}\cdot r$$
because otherwise, assuming that $\lVert uv\rVert \geq \lVert yz\rVert $, for the pair $(u,v)$ and the sequence $e_1=(y,z)$, the left hand side of the inequality in Invariant 2 would be at most
$$2(1+\varepsilon)\cdot\frac{\varepsilon'}{2(1+\varepsilon)}\cdot r + \lVert yz\rVert  \leq (1+\varepsilon')\lVert uv\rVert $$
contradicting the fact that $E$ is consistent with Invariant 2. Thus, given a covering of a ball of radius $R$ with $M$ balls of radius $r'=\frac{\varepsilon'}{2(1+\varepsilon)}\cdot r$, every segment in $E$ has its endpoints in a unique pair of balls, otherwise Invariant 2 will be compromised. Hence, $|E|\leq M^2$. A simple calculation yields a covering with $M<(2(1+\varepsilon)/\varepsilon')^dd^{d/2}(R/r)^d$ balls.
\end{proof}

We can simplify the two invariants by defining a distance function $d_i^*$ over the pairs of vertices,
\begin{definition}
Let $S_i^*$ be a complete weighted graph over the vertices such that the weight of an edge $e$ in $S_i^*$ is defined as
$$
w(e)=\begin{cases}
\lVert e\rVert  & \text{if }e\in S_i \\
(1+\varepsilon)\lVert e\rVert  & \text{if }e\notin S_i
\end{cases}
$$
We define an \emph{extended path} between $u$ and $v$ in $S_i$ as a path between $u$ and $v$ in $S_i^*$ that only uses edges $(y,z)$ where $\lvl(y,z)<\lvl(u,v)$.
We also define the length of an extended path as the sum of its edge weights in $S_i^*$. Finally, we define $d_i^*(u,v)$ as the length of the shortest extended path between $u$ and $v$.
\end{definition}

Using this new distance function we can rephrase the two invariants as follows.

\begin{itemize}
    \item \textbf{Invariant 1.} For every pair $(u,v)\notin S_i$ with $\ind(u,v)=i$, we have $d_i^*(u,v)< (1+\varepsilon)d(u,v)$.
    \item \textbf{Invariant 2.} For every pair $(u,v)\in S_i$, we have $d_i^*(u,v)> (1+\varepsilon')d(u,v)$.
\end{itemize}

It is worth noting that these forms are not exactly equivalent to the previous forms, as we are only considering paths of lower level edges in the definition of $d_i^*$, while a short path in the spanner could potentially contain an edge of the same level. This provides a stronger variation of Invariant 1, which still implies a $1+\varepsilon$ stretch for the spanner. However, this change weakens Invariant 2. But as we will see, a careful addition of the same-level edges can prevent any possible violations of Invariant 2 that could be caused by this new form.

\textbf{Maintaining the invariants.} The quality of our light-weight dynamic spanner depends on the two invariants we introduced above, and an update like a point insertion or removal could cause one of them to break, if not both. Therefore, we establish a procedure that addresses the inconsistencies and enforces the invariants to hold at all times.

The procedure for fixing a violation of Invariant 1 is straightforward: as long as there exists a pair $(u,v)$ that violates Invariant 1 for its corresponding subset $S_i$, add an appropriate potential pair to $S_i$ that connects an ancestor of $u$ to an ancestor of $v$ in the hierarchy $\T$. This resolves the inconsistency for $(u,v)$ if the ancestors are chosen properly, but it might cause other pairs to violate Invariant 2 because of this edge addition. We will prove that if certain criteria are met, there would be no side effect on the same-level pairs and the addition can only result in a constant amortized number of inflicted updates on higher level pairs.

Fixing a violation of Invariant 2, on the other hand, is more tricky. After we remove the violating edge $(u,v)$ from its subset $S_i$, the effect on higher level pairs would be similar to the previous case, but removing $(u,v)$ might cause multiple updates on the same level, which in turn cascade to higher levels. We therefore analyze the removal of $(u,v)$ together with the subsequent additions of same-level edges that aim to fix the incurred violations of Invariant 1, and we prove that a constant amortized bound on the number of inflicted updates on higher level pairs would still hold. We get to the details of our maintenance updates in \cref{sec:maintenance}.

\textbf{Amortized analysis.} We analyze the effects of an update (edge addition and removal) on higher level pairs using a potential function, for each $S_i$ separately. We define our potential function over the potential pairs in $S_i$. The change in the potential function shows how much a pair is close to violating one of the invariants. The higher the potential, the closer the pair is to violating the invariants. This enables us to assign a certain amount of credit to each update, that can be used to pay for the potential change of the updated pair and the affected pairs, which in turn results on an amortized upper bound on the number of edge updates in the future. Therefore, for a potential pair $(u,v)$ 
with index $i$ and
following an update in $S_i$,
\begin{itemize}
    \item if $(u,v)\in S_i$ and $d_i^*(u,v)$ decreases, or
    \item if $(u,v)\notin S_i$, and $d_i^*(u,v)$ increases,
\end{itemize}
we increase the potential of the pair $(u,v)$ to account for its future violation of the invariants.

More specifically, we define the potential function $p_i(u,v)$ of a potential pair $(u,v)$ in $S_i$ as
$$
p_i(u,v)=\begin{cases}
(1+\varepsilon)-d_i^*(u,v)/d(u,v) & \text{if $(u,v)\in S_i$}\\
C_\phi\cdot\left(d_i^*(u,v)/d(u,v)-(1+\varepsilon')\right) & \text{if $(u,v)\notin S_i$ and $\ind(u,v)=i$}
\end{cases}
$$
where $C_\phi>1$ is a positive constant coefficient that we specify later. This implies that if $p_i(u,v)<\varepsilon-\varepsilon'$, then both invariants would hold for the pair $(u,v)$ (in $S_i$). Based on this observation, we define a potential function on $S_i$ in the following way,
$$\Phi_i=\sum_{(u,v)\in \calP_i\cup S_i}p_i(u,v)$$
where $\calP_i$ is the set of potential pairs with index $i$. We simply define the potential of the whole spanner as
$$\Phi = \sum_i \Phi_i$$

We add another term to this potential function later in \cref{sec:light} to account for future edges between the existing nodes.

$$\Phi^* = \Phi + \frac{p_{max}}{2}\cdot\sum_{i=1}^n(D_{max}-\deg_{\calS_1}(v_i))$$

We first prove some bounds on $\Phi$ but we ultimately use the adjusted potential function $\Phi^*$ to prove our amortized bounds on the number of updates. In the remainder of this paper, we specify our sparse and light-weight construction in more details, and we will provide our bounds on the recourse in each case separately.


\section{Sparse spanner}%
\label{sec:sparse}%

In this section, we introduce our dynamic construction for a sparse spanner with constant amortized recourse per point insertion and $\mathcal{O}(\log\Delta)$ recourse per point deletion. We build our spanner on top of a hierarchical clustering that we design early in this section.

Krauthgamer and Lee~\cite{krauthgamer2004navigating} showed how to maintain such hierarchical structures in $\mathcal{O}(\log\Delta)$ update time by maintaining $\varepsilon$-nets. However, this hierarchy is not directly applicable to our case since a point can appear $\log\Delta$ times on its path to root, which would imply a $\mathcal{O}(\log\Delta)$ bound on the degree of the spanner instead of a constant bound. Cloe and Gottlieb~\cite{cole2006searching} improved the update time of this hierarchy to $\mathcal{O}(\log n)$. Gottlieb and Roditty~\cite{gottlieb2008optimal} later introduced a new hierarchical construction with the same update time for their fully-dynamic spanner, which also satisfied an extra close-containment property. Here, we introduce a simpler hierarchy that suits our needs and does not require the close-containment property. Our hierarchy performs constant cluster updates for a point insertion and $\mathcal{O}(\log\Delta)$ cluster updates for a point deletion.

Our hierarchy consists of a pair $(\T, R)$ where $\T$ is a rooted tree of clusters and $R>0$ is a constant.
Every cluster $\C\in\T$ is associated with a center $\cen(\C)\in V$ and a level $\lev(\C)\in\mathbb{Z}$.
The level of a cluster specifies its radius; $\C$ \emph{covers} a ball of radius $R^{\lev(\C)}$ around $\cen(\C)$.
We denote the \emph{parent} of $\C$ in $\T$ by $\prn(\C)$. The root of $\T$, denoted by $\T.root$, is the only cluster without a parent.
Furthermore, the level of a parent is one more than of the child, i.e. $\lev(\prn(\C))=1+\lev(\C)$, for all $\C\in\T$ except the root. A parent must cover the centers of its children.

Besides these basic characteristics, we require our hierarchy to satisfy the separation property at all times. This property states that the clusters at the same level are separated by a distance proportional to their radii,

\begin{definition}[Separation property]
    For any pair of same-level clusters $\C_1,\C_2\in\T$ on level $j$, $$d(\cen(\C_1), \cen(\C_2)) > R^{j}$$
\end{definition}

Each point at the time of insertion creates a single cluster centered at the inserted point, and during the future insertions, might have multiple clusters with different radii centered at it. In fact, each point could have clusters centered at it in at most $\mathcal{O}(\log\Delta)$ levels. At the time of deletion, any cluster that is centered at the deleted point will be removed.

Our clusters are of two types: explicit clusters and implicit clusters. Explicit clusters are the ones we create manually during our maintenance steps. Implicit clusters are the lower level copies of the explicit clusters that exist in the hierarchy even though we do not create them manually. Therefore, if a cluster $\C=(p,l)$ is created in the hierarchy at some point, we implicitly assume clusters $(p,i)$ for $i<l$ exist in the hierarchy after this insertion, and they are included in their corresponding $\T_i$ as well. We maintain the separation property between all clusters, including the implicit ones. We use these implicit clusters for constructing our spanner.

\subsection{Maintaining the hierarchy}
We initially start from an empty tree $\T$ and a constant $R$ that we specify later.

\textbf{Point insertion.} Let $\T_i$ be the set of clusters with level $i$, i.e. $\T_{\lvl(\T.root)}$ only contains the root, $\T_{\lvl(\T.root)-1}$ contains root's children, etc. Upon the insertion of a point $p$, we look for the lowest level (between explicit clusters) $i$ that $p$ is covered in $\T_i$. We insert $\C=(p,i-1)$ into the hierarchy. Since $p$ is covered in $T_i$, we can find a cluster $\C'=(p,i)$ that covers $p$ and assign it as the parent of $\C$ (\cref{alg:hierarchy-insert}).

In the case that $p$ is not covered in any of the levels in $\T$, which we handle by replicating the root cluster from above until it covers the new point, then the insertion happens the same way as before.

\begin{algorithm}[ht]
\caption{Inserting a point to the hierarchy.}\label{alg:hierarchy-insert}
\begin{algorithmic}[1]
\Procedure{Insert-to-Hierarchy}{$\T$, $R$, $p$}
\If {$|\T|=0$}%
\State Add a root cluster $\C=(p, 0)$ to $\T$.%
\State \Return $\C$%
\EndIf%
\State Let $i$ be the lowest level in $\T$.%
\While {$\T_i$ does not cover $p$}%
\State Increase $i$ by 1.%
\If {$i>\lvl{\T.root}$}%
\State Create a new cluster $\C=(\T.root, \lvl(\T.root)+1)$.%
\State Make $\C$ the new root of the hierarchy.%
\State The old root becomes a child of $\C$.%
\EndIf%
\EndWhile%
\State Let $\C'$ be a cluster in $\T_i$ that covers $p$.%
\State Create a cluster $\C=(p,\lvl(\C')-1)$ and add it as a child of $\C'$.%
\EndProcedure%
\end{algorithmic}
\end{algorithm}

The basic characteristics of the hierarchy hold after an insertion. We now show that the separation property holds after the insertion of a new cluster $\C=(p,l)$. Assume, on the contrary, that there exists a cluster $\C'=(q,l)$ that $(\C,\C')$ violates the separation property. $\C$ is inserted on level $l$, thus $p$ is not covered by $\T_l$. According to the assumption, $d(q, p) \leq R^{l}$, meaning that $\C'$ covers $p$. This contradicts the fact that $\T_l$ does not cover $p$ since $\C'\in\T_l$. A similar argument shows that the separation property holds for the implicit copies of $\C$ as well.

\textbf{Point deletion.} Upon the deletion of a point $p$, we remove all the clusters centered at $p$ in the hierarchy. The clusters centered at $p$ create a chain in $\T$ that starts from the lowest level explicit copy of $p$ and ends at the highest level copy. We remove this chain level by level, starting from the lowest level cluster $\C=(p,l)$ that is centered at $p$. Upon the removal of $\C$, we loop over children of $\C$ one by one, and we try to assign them to a new parent. If we find a cluster on level $l+1$ that covers them, then we assign them to that cluster, otherwise we replicate them on one level higher and we continue the process with the remaining children. After we are done with $(p,l)$, we repeat the same process with $(p,l+1)$, until no copies of $p$ exist in the hierarchy (\cref{alg:hierarchy-delete}).

\begin{algorithm}[ht]
\caption{Deleting a point from the hierarchy.}\label{alg:hierarchy-delete}
\begin{algorithmic}[1]
\Procedure{Delete-from-Hierarchy}{$\T$, $R$, $p$}
\State Let $\C=(p,l)$ be the lowest level (explicit) cluster centered at $p$.%
\State Delete $\C$ from $\T$ and mark its children.%
\While {there exists a marked cluster on level $l-1$}%
\State Let $\C'=(q,l-1)$ be a marked cluster.%
\State Find a cluster $\C''$ on level $l$ that covers $q$.%
\If {such cluster exists}%
\State Assign $\C''$ as the parent of $\C'$ and unmark $\C'$.%
\Else%
\State Create $\C''=(q,l)$ and make it the parent of $\C'$.%
\State Mark $\C''$ and unmark $\C'$.%
\EndIf%
\EndWhile%
\If {there still exists a marked cluster in $\T$}%
\State Increase $l$ by one and repeat the while loop above.%
\EndIf%
\EndProcedure%
\end{algorithmic}
\end{algorithm}

Again, the basic characteristics of the hierarchy hold after a deletion. We need to show that the separation property still holds. Immediately after removing the cluster $(p,l)$ the separation property obviously holds. After re-assigning a marked child to another parent the property still holds since no cluster has changed in terms of their center or level. If a marked child is replicated on level $l+1$, it means that there was no cluster covering it on this level, otherwise it would have been assigned as its new parent. Therefore, the separation property holds after the replication on level $l+1$. We will prove more properties of our hierarchy later on when we define the spanner.

\subsection{The initial spanner}
Our initial spanner is a sparse spanner that is defined on the hierarchy $\T$ and it has bounded cluster degree but not bounded point degree. The reason that a bounded degree on the clusters would not imply a bounded degree on the point set is that every point could have multiple clusters centered at it, each of which have a constant number of edges connected to them. This would cause the degree of the point to get as large as $\Omega(\log\Delta)$. Later we will fix this issue by assigning edges connected to large degree points to other vertices.

The initial spanner consists of two types of edges. The first type that we already mentioned before, is the edges that go between clusters of the same level. These edges guarantee a short path between the descendants of the two clusters, similar to a spanner built on a well-separated pair decomposition. And the second type is the parent-child edges, that connect every node to its children. The edge weight between two clusters is the same as the distance between their centers.

We define the spanner formally as follows,
\begin{definition}[Initial spanner]
Let $(\T,R)$ be a hierarchy that satisfies the separation property. We define our sparse spanner $\calS_0$ to be the graph on the nodes of $\T$ that contains the following edges,
\begin{itemize}
	\item Type I. Any pair of centers $p$ and $q$ whose clusters are located on the same level and $d(p,q)\leq \lambda\cdot R^l$ are connected together. Here, $\lambda$ is a fixed constant.
	\item Type II. Any cluster center in $\T$ is connected to the centers of its children in $\T$.
\end{itemize}
\end{definition}

Note that the implicit clusters are also included in this definition. Meaning that if two implicit same-level clusters are close to each other then there would be an edge of type I between them. We show that the spanner $\calS_0$ has a bounded stretch.

\begin{lemma}[Stretch-factor]%
\label{lem:initial-stretch}%
For large enough $\lambda=\mathcal{O}(\varepsilon^{-1})$ the stretch-factor of $\mathcal{S}_0$ would be bounded from above by $1+\varepsilon$.
\end{lemma}
\begin{proof}
Let $p$ and $q$ be two points in the point set, and also let $\C=(p,l)$ and $\C'=(q,l')$ be the highest level clusters in $\T$ that are centered at $p$ and $q$, respectively. By symmetry, assume $l\geq l'$. If $d(p,q) \leq \lambda\cdot R^{l'}$, then there is an edge between the (possibly implicit) cluster $(p,l')$ and $\C'$. This edge connects $p$ and $q$ together, therefore the stretch would be equal to 1 for this pair. If $d(p,q) > \lambda\cdot R^{l'}$, we perform an iterative search for such shortcut edge. Start with $\C=(p,l')$ and $\C'=(q,l')$ and every time that the inequality $d(p,q) \leq \lambda\cdot R^{l'}$ is not satisfied set $\C$ and $\C'$ to their parents and set $l'=l'+1$ and check for the inequality again. We show that the inequality eventually will be satisfied. Let $p_i$ and $q_i$ be the centers of $\C$ and $\C'$ on the $i$-th iteration of this iterative process ($i=1,2,\dotsc$), and let $l'$ have its initial value before any increments. We have $d(p_{i+1}, p_i)\leq R^{l'+i}$ and $d(q_{i+1},q_i)\leq R^{l'+i}$. By the triangle inequality,
$$d(p_{i+1},q_{i+1})\leq d(p_{i+1},p_i) + d(p_i,q_i) + d(q_{i+1},q_i) \leq 2\cdot R^{l'+i} + d(p_i,q_i)$$
Denote the ratio $d(p_i,q_i)/R^{l'+i-1}$ by $x_i$. We have,
$$x_{i+1}\leq 2 + \frac{x_i}{R}$$
Therefore, $x_i$ is roughly being divided by $R$ on every iteration and it stops when $x_i\leq \lambda$. We can easily see that the loop terminates and the value of $x_i$ after the termination would be greater than $\lambda/R$. This particularly shows that the edge between $\C$ and $\C'$ is a long shortcut edge when $\lambda$ is chosen large enough, since its length is more than $\lambda/R$ times the radius of the centers it is connecting.

Now we show that this shortcut edge would be good enough to provide the $1+\varepsilon$ stretch factor for the initial points, $p$ and $q$.
Note that because of the parent-child edges, $p$ can find a path to $q$ by traversing $p_i$s in the proper order and using edge between $p_i$ and $q_i$ and traversing back to $q$. We show that the portion of the path from $p$ to $p_i$ (and similarly from $q$ to $q_i$) is at most $\frac{R^{l'+i}-1}{R-1}$. We prove it only for $p$, the argument for $q$ is similar. Note that if the termination level $l'+i\leq l$ then $p_i=p$ and this path length from $p$ to $p_i$ would be 0, confirming our claim for $p$. Therefore, we assume the termination level is above the level of $p$. The length of the path from $p$ to $p_i$ that only uses type II edges would be at most
$$R^{l+1}+\cdots+R^{l'+i} < \frac{R^{l'+i+1}-1}{R-1}$$
Thus the length of the path from $p$ to $q$ would be at most
$$2\cdot\frac{R^{l'+i+1}-1}{R-1} + d(p_i,q_i)$$
On the other hand, by the triangle inequality,
$$d(p,q)\geq d(p_i,q_i)-2\cdot\frac{R^{l'+i+1}-1}{R-1}$$
Finally, the stretch-factor of this path would be at most
$$\frac{2\cdot\frac{R^{l'+i+1}-1}{R-1} + d(p_i,q_i)}{d(p_i,q_i)-2\cdot\frac{R^{l'+i+1}-1}{R-1}}$$
A simple calculation yields that this fraction is less than $1+\varepsilon$ when $\lambda=2(2+\varepsilon)\varepsilon^{-1}R=\mathcal{O}(\varepsilon^{-1})$.
\end{proof}

Next, we show that the degree of every cluster in $\mathcal{S}_0$ is bounded by a constant. Note that this does not imply a bounded degree on every point, since a point could be the center of many clusters.

\begin{lemma}[Degree bound]%
\label{lem:initial-degree}%
The degree of every cluster in $\mathcal{S}_0$ is bounded by $\mathcal{O}(\varepsilon^{-d})$.
\end{lemma}
\begin{proof}
We first prove that the type I degree of every cluster $\C=(p,l)$ is bounded by a constant. Let $\C'=(q,l)$ be a cluster that has a type I edge to $\C$. This means that $d(p,q)\leq \lambda\cdot R^l$. By the separation property, $d(p,q)>R^{l}$. Thus, by the packing lemma there are at most
$$d^{d/2}\lambda^d=\mathcal{O}(\varepsilon^{-d})$$
type I edges connected to $\C$. The last bound comes from the fact that a choice of $\lambda=\mathcal{O}(\varepsilon^{-1})$ would be enough to have a bounded stretch.

Now we only need to show that the parent-child edges also add at most a constant degree to every cluster, which is again achieved by the packing lemma. Because the children of this cluster are located in a ball of radius $R^l$ around its center, $p$, and they are also pair-wise separated by a distance of at least $R^{l-1}$, we can conclude that the number of children of $\C$ would be upper bounded by $d^{d/2}R^{d}=\mathcal{O}(1)$.
\end{proof}

\textbf{Representative assignment.} So far we showed how to build a spanner that has a bounded degree on each cluster and the desired stretch-factor of $1+\varepsilon$. But this spanner does not have a degree bound on the actual point set and that is a property we are looking for. Here, we show how to reduce the load on high degree points and distribute the edges more evenly so that the bounded degree property holds for the point set as well.

The basic idea is that for every cluster $\C$ in the hierarchy, we pick one of lower level clusters, say $\C'$, to be its \emph{representative} and play its role in the final spanner, meaning that all the spanner edges connecting $\C$ to other clusters will now connect $\C'$ to those clusters after the re-assignment. This re-assignment will be done for every cluster in the hierarchy until every cluster has a representative. Only then we can be certain that the spanner has a bounded degree on the current point set. Since by \cref{lem:initial-degree} the degree of every center is bounded by a constant, we only need to make sure that every point is representing at most a constant number of clusters in the hierarchy.

First, we define the level of a point $p$, denoted by $\lvl(p)$ to be the level of the highest level cluster that has $p$ as its center, i.e. $\lvl(p)=\max_{(p,l)\in\T}l$.

\begin{definition}[Representative assignment]%
\label{def:leaf-rep}%
Let $\T$ be a hierarchy. We define the representative assignment of $\T$ to be a function $\L$ that maps every cluster $\C=(p,l)$ of $\T$ to a point $q$ in the point set such that $l\geq \lvl(q)$ and $d(p,q)\leq R^{l}$. We say $\L$ has bounded repetition $b$ if $|\L^{-1}(q)|\leq b$ for every point $q$.
\end{definition}

Connecting the edges between the representatives instead of the actual centers would give us our bounded-degree spanner.
\begin{definition}[Bounded-degree spanner]%
\label{def:sparse-spanner}%
Define the spanner $\mathcal{S}_1$ to be the spanner connecting the pair ($\L(\C),\L(\C'))$ for every edge $(\C,\C')\in\mathcal{S}_0$.
\end{definition}

Now we show that this re-assignment of the edges would not affect the stretch-factor and the degree bound significantly if the clusters are small enough, or equivalently, $\lambda$ is chosen large enough.
\begin{lemma}[Stretch-factor]%
\label{lem:sparse-stretch}%
For large enough $\lambda=\mathcal{O}(\varepsilon^{-1})$ and any representative assignment $\L$ the stretch-factor of $\mathcal{S}_1$ would be bounded from above by $1+\varepsilon$.
\end{lemma}
\begin{proof}
The proof works in a similar way to the proof of \cref{lem:initial-stretch}. A shortcut edge would still provide a good path between two clusters even after its end points are replaced by their representatives. The path from a $p$ to $p_i$ will be doubled at most since a representative could be as far as a child from the center of a cluster. Therefore, the stretch-factor of the path between $p$ and $q$ will be
$$\frac{4\cdot\frac{R^{l'+i+1}-1}{R-1} + d(p_i,q_i)}{d(p_i,q_i)-4\cdot\frac{R^{l'+i+1}-1}{R-1}}$$
Again, this fraction is less than $1+\varepsilon$ when $\lambda=4(2+\varepsilon)\varepsilon^{-1}R=\mathcal{O}(\varepsilon^{-1})$.
\end{proof}

To construct a bounded-repetition representative assignment we pay attention to the neighbors of lower level copies of a cluster. Let $\C=(p,l)$ be a cluster that we want to find a representative for. As we mentioned before, $(p,l')$ exists in the hierarchy for all $l'<l$. If $l'$ is small enough, i.e. $l' < l - \log_R \lambda$, then the neighbors of $(p,l')$ will be located within a distance $\lambda\cdot R^{l'}=R^l$ of $p$, making them good candidates to be a representative of $\C$. Therefore, having more neighbors on lower levels means having more (potential) representatives on higher levels. This is how we assign the representatives.

We define a chain to be a sequence of clusters with the same center that form a path in $\T$. We divide a chain into blocks of length $\log_R \lambda$. The best way to do this so that maintaining it dynamically is easy is to index the clusters in a chain according to their levels and gather the same indices in the same block. We define the block index of a cluster in a chain to be $\lfloor l/\log_R \lambda\rfloor$, where $l$ is the level of the cluster. The clusters in a chain that have the same index form a block.

The first observation is that if we are given two non-consecutive blocks in the same chain, we can use the neighbors of the lower level block as representatives of the higher level block. This is the key idea to our representative assignment, which we call \emph{next block assignment}. In this assignment, we aim to represent higher level points with lower level points.
Let $p$ be a point and $P_1,P_2,\dotsc,P_k$ be all the blocks of the chain that is centered at $p$ in $\T$, ordered from top to bottom (higher level blocks to lower level blocks). We say a block is empty if the clusters in the block have no neighbors in $\T$. We say the block is non-empty otherwise. We make a linked list $\L_0$ of all the even indexed non-empty blocks, and a separate linked list $\L_1$ for all the odd indexed non-empty blocks. For every element of $\L_0$ we pick an arbitrary neighbor cluster of its block in $\L_0$ (because the blocks are non-empty such neighbors exists), and we assign that neighbor to be the representative of the clusters in that element. More specifically, let $B_i$ be a block in $\L_0$, and let $B_{i+1}$ be the next block in $\L_0$. Let $\C$ be an arbitrary cluster in $B_{i+1}$ that has a neighbor. This cluster exists, since $B_{i+1}$ is a non-empty block. Let $q$ be the center of a neighbor of $\C$. We assign $\L(\C') = q$ for all $\C'\in B_i$. The same approach works for $\L_1$. This assigns a representative to every block in the chain, except the last block in $\L_0$ and $\L_1$. We assign $p$ itself to be the representative of the clusters in these blocks.

Now we show that this assignment has bounded repetition. First, we show that our assignment only assigns lower level points to be representatives of higher level points.

\begin{lemma}\label{lem:rep-order}
Let $p$ and $q$ be two points in the point set and let $\lvl(p) > \lvl(q)$ . In the next block assignment $q$ would never be represented by $p$.
\end{lemma}
\begin{proof}
Assume, on the contrary, that $q$ is represented by $p$. Therefore, there exists two same-parity cluster blocks in the chain centered at $q$ that a cluster centered at $p$ is connected to the lower block. Let $\C=(p,l)$ and $\C'=(q,l')$ be the highest clusters centered at $p$ and $q$, respectively. Since the connection between $p$ and $q$ is happening somewhere on the third block or lower on the chain centered at $q$, we can say that $d(p,q)< \lambda\cdot R^{l'-\log_R\lambda} = R^{l'}$. This means that the separation property does not hold for the lower level copy of $\C$, $(p,l')$, and $\C'$, which is a contradiction.
\end{proof}

Now that we proved that points can only represent higher level points in our assignment, we can show the bounded repetition property.

\begin{lemma}[Bounded repetition]\label{lem:rep-bound}
The next block assignment $\L$ described above has bounded repetition.
\end{lemma}
\begin{proof}
We show that every point represents at most a constant number of clusters. First, note that the two bottom clusters of the two block linked lists have a constant number of clusters in them (to be exact, $2\log_R\lambda$ clusters maximum). So we just need to show that the number of other clusters that are from other chains and assigned to the point are bounded by a constant. Let $p$ be an arbitrary point and let $\C=(p,l)$ be the highest level cluster centered at $p$. According to the previous lemma, any point $q$ that has a cluster $\C'=(q,l')$ that $\L(\C')=p$ must have a higher level than $p$. Therefore, there exists a lower level copy of $q$ on level $l$. Also, the distance between $p$ and $q$ is bounded by $\lambda\cdot R^{l}$ since $p$ and $q$ are connected on a level no higher than $l$ (remember that we only represent our clusters with their lower level neighbors). Now we can use the packing lemma, since all such points $q$ have a cluster centered at them on level $l$ and therefore separated by a distance of $R^l$. By the packing lemma, the number of such clusters would be bounded by $d^{d/2}\lambda^d$ such points. So the repetition is at most $b=d^{d/2}\lambda^d+2$.
\end{proof}

\begin{corollary}
The spanner $\calS_1$ has bounded degree.
\end{corollary}

\subsection{Maintaining the spanner}

So far we showed $\calS_1$ has bounded stretch and bounded degree. Here we show that we can maintain $\calS_1$ with $\mathcal{O}(1)$ amortized number of updates after a point insertion and $\mathcal{O}(\log\Delta)$ amortized number of updates after a point deletion. We know how to maintain the hierarchy from earlier in this section. Therefore, we just explain how to update the spanner, which includes maintaining our representative assignments dynamically.

\textbf{Point insertion.} We prove the amortized bound by assigning credits to each node, and using the credit in the future in the case of an expensive operation. Let $D_{max}$ be the degree bound we proved for $\calS_1$. When a new point is added to the spanner, we assign $D_{max}$ credits to it.

We analyze the edge addition and removals that happen after the insertion of a point $p$ in the spanner. Note that although only one explicit cluster is added to $\T$ after the insertion, there might be many new edges between the implicit (lower level) copies of the new cluster and other clusters that existed in $\T$ beforehand. We need to show that these new edges do not cause a lot of changes on the spanner after the representative assignment phase.

First, we analyze the effect of addition of $p$ on points $q$ that $\lvl(p) > \lvl(q)$. Similar to the proof of \cref{lem:rep-order}, we can show that any edges between the chain centered at $p$ and the chain centered at $q$ will be connected to the top two cluster blocks of the chain centered at $q$. This means that these edges will have no effect on the assignment of other clusters in the chain centered at $q$, because each non-empty block is represented by some neighbor of the next non-empty same-parity block, and the first two blocks, whether they are empty or not, will not have any effect on the rest of the assignment. Therefore, no changes will occur on the representatives of $q$ and therefore the edges that connect these representatives together will remain unchanged.

The addition of $p$ as we mentioned, would cause the addition of some edges in the spanner $\calS_1$, that we pay for using the constant amount of credit stored on the endpoints of those edges. Therefore, we are not spending more than constant amount of amortized update for this case.

Second, we analyze the effect of addition of $p$ on points $q$ that $\lvl(p) \leq \lvl{q}$. The outcome is different in this case. Similar to the previous case we can argue that any edge between the chain centered at $p$ and the chain centered at $q$ must be connected to the top two blocks of the chain centered at $p$, but they could be connected to anywhere relative to the highest cluster centered at $q$. This means that they could add a non-empty block in the middle of the chain centered at $q$. If this happens, then the assignment of the previous non-empty same-parity block changes and also the new non-empty block will have its own assignment. This translates into a constant number of changes (edge additions and removals) on the spanner $\calS_1$ per such point $q$. We earlier in \cref{lem:rep-bound} proved that there is at most a constant number of such clusters. This shows that there would be at most a constant number of changes on the spanner $\calS_1$ from higher level points.

Finally, we can conclude that overall the amortized recourse for insertion is bounded by a constant, since in the first case we could pay for the changes using the existing credits, and in the second case we could pay for the changes from our pocket.

\textbf{Point deletion.} After a point deletion, all the clusters centered at that point will be removed from the hierarchy, and a set of replication to higher levels would happen to some clusters to fix the hierarchy after the removal. It is easy to see that the number of cluster changes (including removal and replication) would be bounded by a constant. Each cluster change would also cause a constant number of changes on the edges of the spanner $\calS_0$. Note that a cluster removal can introduce an empty block to at most a constant number of higher level points and a cluster replication can also introduce an empty block to at most a constant number of higher level points. Therefore, the changes on the representative assignments would be bounded by a constant after a single cluster update. Since we have at most $\mathcal{O}(\log\Delta)$ levels in the hierarchy, each of which having at most a constant number of cluster updates, overall we would have at most $\mathcal{O}(\log\Delta)$ number of edge changes on $\calS_1$. After the removal, we assign full $D_{max}$ credit to any node that is impacted by the removal. This would make sure we have enough credits for the future additions.

\section{Light spanner}%
\label{sec:light}%

So far we introduced our hierarchy and how to maintain it under point insertions and removals, and also how to create a spanner on top of the hierarchy and how to make it sparse with representative assignments. We also studied how our sparse spanner changes under point insertions and point removals and we bounded the amortized number of updates per insertion to a constant, and the bound for the number of updates per deletion to $\mathcal{O}(\log\Delta)$.

In this section, we introduce our techniques for maintaining a light spanner that has a constant lightness bound on top of all the properties we had so far. In our main result in this section we show that maintaining the lightness in our case is not particularly harder than maintaining the sparsity, meaning that it would not require asymptotically more changes than a sparse spanner would.

We ultimately want to select a subset of the edges of our sparse spanner that are light and preserve the bounded stretch to achieve a bounded degree light spanner. For this purpose, we introduce a set of maintenance updates that we perform after point insertion and removal. These maintenance updates aim to improve the weight of the spanner in iterations. In each iteration, we look at all the edge buckets of our spanner, and we search for an edge that does not satisfy the leapfrog property for certain constants. If no such edge is found in the buckets then the spanner's lightness is already bounded by a constant. If found, such an edge would be deleted from the bucket and removed from the spanner. The removal of this edge could cause the bounded stretch property to not hold for some other pairs. We take one such pair and we add the edge between the two points. Now the addition of the new edge could cause the appearance of some pairs that violate the leapfrog property and therefore increase the lightness. We then repeat this loop of removal and addition again until we reach a bounded lightness bounded stretch spanner. We show in this section that the number of iterations we need to reach a bounded degree bounded stretch spanner is proportional to the number of edges that we changed since our last bounded degree bounded stretch state. Therefore, if we only change an amortized constant number of edges to reach a state, then an amortized constant number of updates would be enough to make that state stable again.

We first analyze the effect of point insertion or deletion on the potential functions we defined earlier in \cref{sec:prelim}. Then we introduce our maintenance updates and we show our bounds on the recourse of a light spanner.

\subsection{Bounding the potential function}%
\label{sec:potential-function}%
In this section we analyze the behavior of our potential functions, after a point insertion and a point deletion. These bounds will later help us prove the amortized bounds on the recourse. As we defined in \cref{sec:prelim}, the potential function $p_i(u,v)$ on a potential pair $(u,v)$ in a bucket $S_i$ is equal to
$$
p_i(u,v)=\begin{cases}
(1+\varepsilon)-d_i^*(u,v)/d(u,v) & \text{if $(u,v)\in S_i$}\\
C_\phi\cdot\left(d_i^*(u,v)/d(u,v)-(1+\varepsilon')\right) & \text{if $(u,v)\notin S_i$ and $\ind(u,v)=i$}
\end{cases}
$$
And the overall potential function on a bucket is defined as
$$\Phi_i=\sum_{(u,v)\in \calP_i\cup S_i}p_i(u,v)$$
where $\calP_i$ is the set of potential pairs with index $i$. And we defined a potential function on the whole spanner as
$$\Phi = \sum_i \Phi_i$$

\textbf{Single edge update.} We start with a simple case of bounding the potential function after a single edge insertion, then we consider a single edge deletion, and finally we extend our results to point insertions and deletions. We assume the pair that we insert to or delete from the spanner is an arbitrary pair from the set of potential pairs, because we only deal with potential pairs in our light spanner.

First, we consider a single edge insertion. We divide the analysis into two parts: the effect of the insertion of the potential pair onto the same level potential pairs, and the effect of the insertion onto higher level potential pairs. Recall that the level of a pair was defined in \cref{sec:prelim}.

We show that the edges of the same level satisfy a separation property, meaning that two edges in the same bucket cannot have both their endpoints close to each other.
\begin{lemma}[Edge separation]%
\label{lem:edge-separation}%
Let $(u,w)$ and $(y,z)$ be two potential pairs in the same bucket. Assuming that $(u,w)$ and $(y,z)$ are not representing clusters from the same pair of chains in $\T$,
$$\max\{d(u,y),d(w,z)\}>\frac{1}{\lambda^2\cdot c}\max\{d(u,w),d(y,z)\}$$
\end{lemma}
\begin{proof}
Note that the constraint on not connecting the same pair of chains in the lemma is necessary, because in our sparse spanner construction, it is possible that two points are connected on two different levels on two different pairs of clusters. These two edges could potentially go into different non-empty blocks and get assigned different representatives and cause two parallel edges between two neighborhoods. While this is fine with sparsity purposes as long as there is at most a constant number of such parallel edges, we do not want to have them in our light spanner since they will make the analysis harder. Therefore, we assume that the edges are not connecting clusters centered at the same pair of points.

Next we show that these two pairs are from two cluster levels that are not far from each other. Let $(u,w)$ be an edge on level $l$ of the hierarchy and $(y,z)$ be an edge on level $l'$ of the hierarchy. Without loss of generality, assume that $l\geq l'$. We know that the potential pairs connect same level clusters together. Therefore, the length of $(u,w)$ could vary between $R^l$ and $\lambda\cdot R^l$. A similar inequality holds for $(y,z)$. Thus the ratio of the length of the two would be at least $\lambda^{-1}R^{l-l'}$. Also, if $C$ is chosen large enough it is clear that the two edges must have the same index as well, otherwise the length ratio of $C$ between the two edges would make their endpoints very far from each other. Thus, the edges belong to the same bucket and index, meaning that the length of their ratio is at most $c$. So,
$$\lambda^{-1}R^{l-l'} < c$$
Now, the separation property on level $l'$ between the clusters that these two edges are connecting to each other states that
$$\max\{d(u,y),d(w,z)\}\geq R^{l'}>\frac{R^l}{\lambda\cdot c}$$
Also according to earlier in this proof, $R^l\geq d(u,w)/\lambda$. Thus,
$$\max\{d(u,y),d(w,z)\}>\frac{d(u,w)}{\lambda^2\cdot c}=\frac{1}{\lambda^2\cdot c}\max\{d(u,w),d(y,z)\}$$
\end{proof}

Now using this lemma we show that the insertion of a potential pair will not cause any violations of Invariant 2 on the same level.

\begin{lemma}%
\label{lem:single-ins-same}%
Let $(u,w)$ be a potential pair that is inserted to $S_i$ where $i=\ind(u,w)$. If $d_i^*(u,w)>(1+\varepsilon')d(u,w)$, then the insertion of $(u,w)$ results in no violations of Invariant 2 on same or lower level edges, assuming that $c^{-1}(1+\lambda^{-2})\geq 1+\varepsilon'$.
\end{lemma}
\begin{proof}
It is clear that $(u,w)$ cannot participate in a shortest-path (in $S_i^*$) for any of the lower level pairs, so adding it does not affect any of those pairs. Also adding $(u,w)$ would not violate Invariant 2 for the pair itself because of the assumption $d_i^*(u,w)>(1+\varepsilon')d(u,w)$. Thus we only need to analyze the other same level edges.

So let $(y,z)$ be a same-level edge in $S_i$. If one of $(u,w)$ or $(y,z)$ use the other one in its shortest extended path (in $S_i^*$), then by \cref{lem:edge-separation}, the length of the path would be at least
$$\min\{d(u,w),d(y,z)\}+\max\{d(u,y),d(w,z)\}>\min\{d(u,w),d(y,z)\}+\frac{1}{\lambda^2\cdot c}\max\{d(u,w),d(y,z)\}$$
We also know, from the assumption, that $(u,w)$ and $(y,z)$ are same-level edges in $S_i$, so $c^{-1}<d(u,w)/d(y,z)<c$. Therefore, the stretch of the path would be at least
$$\frac{\min\{d(u,w),d(y,z)\}+\max\{d(u,y),d(w,z)\}}{\max\{d(u,w),d(y,z)\}}>c^{-1}(1+\lambda^{-2})\geq1+\varepsilon'$$
Thus the stretch of the path is more than $1+\varepsilon'$, which shows that this addition would not violate Invariant 2 for any of the two pairs,
even though the paths of same level edges are excluded in $d_i^*(u,w)$.
\end{proof}

Note that satisfying the condition in \cref{lem:single-ins-same} is easy. We first choose large enough $\lambda$ to have a fine hierarchy, then we choose $c$ small enough that $c < 1+\lambda^{-2}$, then we choose $\varepsilon'=c^{-1}(1+\lambda^{-2})-1$. Now we show that the potential change on higher level potential pairs would be bounded by a constant after the insertion of $(u,w)$.

\begin{lemma}%
\label{lem:single-ins-high}%
Let $(u,w)$ be a potential pair that is inserted to $S_i$ where $i=\ind(u,w)$. The insertion of $(u,w)$ results in at most
$$\frac{C_3}{c^k-1}$$
potential increase on higher level potential pairs in $S_i$, where
$$C_3=\varepsilon(1+\varepsilon)^dc^{d+1}C_1$$
is a constant (and $k$ is the number of buckets).
\end{lemma}
\begin{proof}
Let $(y,z)$ be an edge of level $j'>j$ in $S_i$ whose $d_i^*$ is decreased by the addition of $(u,w)$. Thus the shortest extended path between $y$ and $z$ in $S_i^*$ passes through $(u,w)$. Denote this path by $P_i^*(y,z)$. Before the addition of $(u,w)$, the length of the same path in $S_i^*$ was at most $\lVert P_i^*(y,z)\rVert + \varepsilon d(u,w)$. Hence, $\Delta d_i^*(y,z)\geq -\varepsilon d(u,w)$, and the potential change of this edge would be
$$\Delta p_i(y,z)=\frac{-\Delta d_i^*(y,z)}{d(y,z)}\leq \frac{\varepsilon d(u,w)}{d(y,z)}\leq \varepsilon c^{k(j-j')+1}$$

In the next step, we bound the number of such $(y,z)$ pairs. Let $r$ be the minimum length of such edge in level $j'$. Both $y$ and $z$ must be within $(1+\varepsilon)cr$ Euclidean distance of $u$ (and $w$), otherwise the edge $(u,w)$ would be useless in $(y,z)$'s shortest path in $S_i^*$. Thus, all such pairs are located in a ball $B(u,(1+\varepsilon)cr)$, and according to \cref{lem:edge-packing}, there would be at most
$$C_2 = (1+\varepsilon)^dc^dC_1$$
number of them.

Thus, the overall potential change on level $j'$ would be upper bounded by $C_2 \varepsilon c^{k(j-j')+1}$. Summing this up over $j'>j$, the overall potential change on higher level pairs would be at most
$$\Delta \Phi_i < \sum_{j'>j} \varepsilon C_2 c^{k(j-j')+1}=\frac{C_3}{c^k-1}$$
where $C_3=\varepsilon C_2c$.
\end{proof}

Now we analyze the removal of a potential pair from a bucket. The difference with the removal is that it could cause violations of Invariant 1 on its level. Therefore, we analyze a removal, together with some subsequent edge insertions that fix any violations of Invariant 1 on the same level.

\begin{definition}[Edge removal process]
Let $(u,w)$ be a potential pair that is located in $S_i$ where $i=\ind(u,w)$. We define the single edge removal process on $(u,w)$ to be the process that deletes $(u,w)$ from $S_i$ and fixes the subsequent violations of Invariant 1 on the same level by greedily picking a violating pair, and connecting its endpoints in $S_i$, until no violating pair for Invariant 1 is left.
\end{definition}

We analyze the effect of the edge removal process in the following two lemmas,
\begin{lemma}%
\label{lem:single-rem-same}%
Let $(u,w)$ be a potential pair that does not violate Invariant 1 ($d_i^*(u,w)<(1+\varepsilon)d(u,w)$) and is deleted from $S_i$ ($i=\ind(u,w)$), using the edge removal process. The deletion of $(u,w)$ together with these subsequent insertions results in no violations of Invariant 1 or Invariant 2 on same or lower level edges, assuming that $c^{-1}(1+\lambda^{-2})\geq 1+\varepsilon'$.
\end{lemma}
\begin{proof}
It is clear that $(u,w)$ cannot participate in a shortest-path (in $S_i^*$) for any of the lower level pairs, so deleting it does not affect any of those pairs. Also, every same level pair that violates Invariant 1 is fixed after the insertion of subsequent edges. Therefore, we just need to show there are no violations of Invariant 2 after these changes. This is also clear by \cref{lem:single-ins-same}, because we are only inserting edges $(y,z)$ that that violate Invariant 1, i.e. $d_i^*(y,z)>(1+\varepsilon)d(y,z)>(1+\varepsilon')d(y,z)$, meaning that the assumption of the lemma holds in this insertion.
\end{proof}

We show a similar bound as edge insertion on the effect of the edge removal process on higher level pairs.
\begin{lemma}%
\label{lem:single-rem-high}%
Let $(u,w)$ be a potential pair that is deleted from to $S_i$ where $i=\ind(u,w)$. The edge removal process on $(u,w)$ results in at most 
$$\frac{C_5}{c^k-1}$$
potential increase on higher level potential pairs in $S_i$, for some constant $C_5$ that depends on $\varepsilon$, $\varepsilon'$, and $c$.
is a constant.
\end{lemma}
\begin{proof}
The edge removal process can be divided into two phases. The deletion of $(u,w)$, and the insertion of the subsequent pairs. First, we show that the potential increase after the edge deletion is bounded. Let $(y,z)$ be an edge of level $j'>j$ in $S_i$ whose $d_i^*$ is increase by the deletion of $(u,w)$. Thus the shortest extended path between $y$ and $z$ in $S_i^*$ passes through $(u,w)$. Denote this path by $P_i^*(y,z)$. After the removal of $(u,w)$, the length of the same path in $S_i^*$ is at most $\lVert P_i^*(y,z)\rVert + \varepsilon d(u,w)$. Hence, $\Delta d_i^*(y,z)\leq \varepsilon d(u,w)$, and the potential change of this edge would be
$$\Delta p_i(y,z)=\frac{\Delta d_i^*(y,z)}{d(y,z)}\leq \frac{\varepsilon d(u,w)}{d(y,z)}\leq \varepsilon c^{k(j-j')+1}$$
Again, the number of such $(y,z)$ pairs is bounded by $$C_2 = (1+\varepsilon)^dc^dC_1$$ according to \cref{lem:edge-packing}. Thus, the overall potential change on level $j'$ would be upper bounded by $C_2 \varepsilon c^{k(j-j')+1}$. Summing this up over $j'>j$, the overall potential change on higher level pairs would be at most
$$\Delta \Phi_i < \sum_{j'>j} \varepsilon C_2 c^{k(j-j')+1}=\frac{C_3}{c^k-1}$$
where $C_3=\varepsilon C_2c$.

Now, the number of subsequent edge insertions would also bounded by a constant. Because in order for an inserted pair $(y,z)$ to violate Invariant 1 after the deletion of $(u,w)$, $u$ and $w$ must be within a distance $c(1+\varepsilon)d(u,w)$, otherwise the edge $(u,w)$ would be useless in their shortest-path. Also since they satisfy Invariant 2, we conclude from \cref{lem:edge-packing} that the number of such pairs is bounded by a constant. Denote this bound by $C_4$. Then the potential on higher level pairs from the insertions of $C_4$ pairs on the same level would be at most $C_3C_4/(c^k-1)$.

Overall, the potential increase on higher level pairs from the edge removal process will be $C_5/(c^k-1)$ where $C_5=C_3(C_4+1)$.
\end{proof}

\textbf{Adjusted potential function.} We have one last step before analyzing the potential function after a point insertion and a point deletion. We need to slightly adjust the potential function to take into account future edges that might be added between the existing points because of a new point. As we saw in \cref{sec:sparse}, a new point can have a large degree in $\calS_0$ due to its implicit clusters in multiple levels of the hierarchy. We handled this by assigning these edges to nearby representatives and we proved a constant degree bound on $\calS_1$. But this still would mean adding a point could increase the potential function by $\Omega(\log\Delta)$ since logarithmic number of edges could be added to the sparse spanner. We fix this issue in our potential function by taking into account all the future edges that can be incident to a point. Our adjusted potential function on the whole spanner, denoted by $\Phi^*$, has an extra term compared to the previous potential function $\Phi$,
$$\Phi^* = \Phi + \frac{p_{max}}{2}\cdot\sum_{i=1}^n(D_{max}-\deg_{\calS_1}(v_i))$$
$\deg_{\calS_1}(v_i)$ is the degree of the $i$-th point (in any fixed order, e.g. insertion order) in the sparse bounded degree spanner $\calS_1$, and
$$p_{max}=\max\{1+\varepsilon, C_\phi(\varepsilon-\varepsilon')\}$$
is the maximum potential value a potential pair can have in its own bucket given the fact that it does not violate Invariant 1. Note that the first term is the maximum of the potential of any pair if its edge is present in the bucket and the second term is the maximum potential of the pair if its edge is absent from the bucket and it is not violating Invariant 1. We will later see why the assumption that Invariant 1 holds for such pairs is fine. But this extra term in the potential function will be used to cover the potential $p_i$ of the extra potential pairs added by the new point.

\subsection{Maintaining the light spanner}
We are finally ready to introduce our techniques for maintaining a light spanner under a dynamic point set. For point insertion, we select a subset of edges added in the sparse spanner to be present in the light spanner. We show that the potential increase on $\Phi^*$ after inserting the new point would be bounded by a constant. Then we perform the same analysis for point deletion and we show that the potential increase is bounded by $\mathcal{O}(\log\Delta)$. In the last part of this section we introduce our methods for iteratively improving the weight of the spanner by showing an algorithm that decreases the potential function by a constant value in each iteration. This concludes our results on the recourse for point insertion and point deletion.

\textbf{Point insertion.} Following a point insertion for a point $p$, we insert $p$ into the hierarchy and we update our sparse spanner $\calS_1$. There are at most a constant number of pairs whose representative assignment has changed, we update these pairs in the light spanner as well. Meaning that if they were present in the light spanner, we keep them present but with the new endpoints, and if they were absent, we keep them absent. Besides the re-assignments, there could be some (even more than a constant) edge insertions into the sparse spanner, but the degree bound of $D_{max}$ would still hold on every point. We greedily pick one new edge at a time that its endpoints violate Invariant 1 in the light spanner, and we add that edge to the light spanner. (\cref{alg:light-insert})

\begin{algorithm}[ht]
\caption{Inserting a point to the light spanner.}\label{alg:light-insert}
\begin{algorithmic}[1]
\Procedure{Insert-to-Light-Spanner}{$p$}
\State Insert $p$ into the hierarchy $\T$.%
\State Make the required changes on the sparse bounded degree spanner $\calS_1$.%
\For {any pair $(u,w)$ with updated representative assignment}%
\State Update the endpoints of the edge in the light spanner.%
\EndFor%
\For {any edge $(u,w)$ added to the sparse spanner}%
\If {Invariant 1 is violated for this pair on the light spanner}%
\State Add $(u,w)$ to the light spanner (to its own bucket).%
\EndIf%
\EndFor%
\EndProcedure%
\end{algorithmic}
\end{algorithm}

We now analyze the change in the potential function after performing this function following a point insertion.

\begin{lemma}%
\label{lem:light-ins-potential}%
The procedure $\Call{Insert-to-Light-Spanner}{}$ adds at most a constant amount to $\Phi^*$.
\end{lemma}
\begin{proof}
Note that at most a constant number of edges will go through a representative assignment change. Each representative change can be divided into removing the old pair and adding the new one. Each removal will increase the potential of at most a constant number of pairs on any same or higher level pairs. This would sum up to a constant amount as we saw earlier in \cref{lem:single-rem-same} and \cref{lem:single-rem-high}. Also, inserting the updated pairs would also sum up to a constant amount of increase in the potential function as we saw in \cref{lem:single-ins-same} and \cref{lem:single-ins-high}.

For the edge insertions however, we will get help from the extra term in our potential function. Note that any extra edge that is added between any two points that existed before the new point will increase both of their degrees by 1 and therefore, decrease the term 
$$p_{max}\cdot\sum_{i=1}^n(D_{max}-\deg_{\calS_1}(v_i))$$
by $p_{max}$. On the other hand, the new pair will either be added to the light spanner or will satisfy Invariant 1 if not added. Thus, its potential will be at most $1+\varepsilon$ in the first case, and at most $C_\phi(\varepsilon-\varepsilon')$ in the second case. In any case, the potential of the new pair is not more than $p_{max}$, and hence $\Phi^*$ will not increase due to the addition of the new pair.

Lastly, the new point will introduce a new term
$p_{max}\cdot(D_{max}-\deg_{\calS_1}(v_{n+1}))$
in $\Phi^*$ which would also be bounded by a constant. Overall, the increase in $\Phi^*$ will be bounded by a constant.
\end{proof}

This lemma by itself does not provide an upper bound on the number of inflicted updates. However, later in \cref{sec:maintenance}, when we analyze our maintenance updates, we use this lemma to prove that the amortized number of edge updates would be bounded by a constant.

\textbf{Point deletion.} Following a point deletion, we perform the deletion on the hierarchy and update the sparse spanner accordingly. This would cause at most $\mathcal{O}(\log\Delta)$ potential pairs to be deleted from or inserted into the spanner. The procedure on the light spanner is simple in this case. We add all the inserted pairs to the light spanner, and we remove the removed pairs from the light spanner if they are present.

\begin{algorithm}[ht]
\caption{Deleting a point from the light spanner.}\label{alg:light-del}
\begin{algorithmic}[1]
\Procedure{Delete-from-Light-Spanner}{$p$}
\State Delete $p$ from the hierarchy $\T$.%
\State Make the required changes on the sparse bounded degree spanner $\calS_1$.%
\For {any pair $(u,w)$ removed from the sparse spanner}%
\State Remove $(u,w)$ from the light spanner if present.%
\EndFor%
\For {any pair $(u,w)$ added to the sparse spanner}%
\State Add $(u,w)$ to the light spanner.%
\EndFor%
\For {any pair $(u,w)$ with updated representative assignment}%
\State Update $(u,w)$ in the light spanner as well.%
\EndFor%
\EndProcedure%
\end{algorithmic}
\end{algorithm}

\begin{lemma}%
\label{lem:light-del-potential}%
The procedure $\Call{Delete-from-Light-Spanner}{}$ adds at most $\mathcal{O}(\log\Delta)$ to $\Phi^*$.
\end{lemma}
\begin{proof}
The number of edges updated on every level of hierarchy after a point removal is bounded by a constant. Therefore, the total number of changes would be bounded by $\mathcal{O}(\log\Delta)$. Each change would cause $\Phi^*$ to increase by at most $p_{max}$. Thus, the total increase is bounded by $\mathcal{O}(\log\Delta)$.
\end{proof}

\subsection{Maintenance updates}%
\label{sec:maintenance}%

As we saw earlier in this section, following a point insertion and removal many edge updates happen on the light spanner, and we did not check for the invariants to hold after these changes. Maintaining Invariant 1 and Invariant 2 is crucial for the quality of our spanner. Here, we show how we can maintain both invariants following a point insertion and deletion. We also complete our amortized bounds on the number of updates required to make the spanner consistent with the two invariants.

Our maintenance updates are of two different types, each designed to fix the violations of one invariant. Whenever a violation of Invariant 1 occurs for a potential pair $(v,w)$, we fix the violation by simply adding the edge between $v$ and $w$ to its corresponding $S_i$. This fixes the violation for this pair, and all pairs of descendants of the two clusters that $v$ and $w$ represent. We will show that this change will decrease the value of the potential function by a constant amount.

Fixing a violation of Invariant 2 on the other hand is not that simple. Removing the edge between $v$ and $w$ might cause multiple violations of Invariant 1 on the same level. As we discussed before, we address this issue by fixing the same-level violations of Invariant 1 first through adding edges between $v$'s neighborhood and $w$'s neighborhood. Then we show that the removal of $(v,w)$ together with these additions would lower the value of the potential function by a constant amount.

Our maintenance approach is simple, as long as there exists a potential pair on any $S_i$ that violates either of the two invariants, we perform the corresponding procedure to enforce that invariant for that pair. The fact that the potential function decreases by a constant amount after each fix is the key to our amortized analysis on the number of maintenance updates to reach a spanner with bounded degree and bounded lightness.

\textbf{Fixing a violation of Invariant 1.} In our first lemma in this section, we show that fixing a violation of Invariant 1 in the way that we mentioned above, would decrease the value of the potential function on each $S_i$.

\begin{lemma}%
\label{lem:fix-inv1}%
Let $(v,w)$ be a potential pair with $\ind(v,w)=i$ that violates Invariant 1, i.e. $d_i^*(v,w)/d(v,w)> 1+\varepsilon$. Also, assume that
$$k\geq \log_c\left(1+\frac{C_3}{(C_\phi-1)(\varepsilon-\varepsilon')}\right)$$
Then adding the edge $(v,w)$ to $S_i$ decreases the overall potential $\Phi_i$ of $S_i$ by at least $(\varepsilon-\varepsilon')$.
\end{lemma}
\begin{proof}
Note that adding $(v,w)$ would have no effect on the potential of the lower level or same level potential pairs, due to the definition of $d_i^*$. We know from \cref{lem:single-ins-high} that adding $(v,w)$ to $S_i$ would increase the potential on higher level pairs by at most $C_3/(c^k-1)$. Also, the potential of the pair itself before the addition is
$$p_i(v,w)=C_\phi\cdot\left(\frac{d_i^*(v,w)}{d(v,w)}-(1+\varepsilon')\right)$$
On the other hand, after the addition,
$$p_i(v,w)=(1+\varepsilon)-\frac{d_i^*(v,w)}{d(v,w)}$$
Therefore,
$$\Delta p_i(v,w)=(\varepsilon-\varepsilon')+(C_\phi+1)\left(1+\varepsilon'-\frac{d_i^*(v,w)}{d(v,w)}\right)$$
We know by the assumption that the stretch of the shortest extended path between $v$ and $w$ would be more than $1+\varepsilon$, since $(v,w)$ is violating Invariant 1. Therefore,
$$1+\varepsilon'-\frac{d_i^*(v,w)}{d(v,w)}< -(\varepsilon-\varepsilon')$$
Thus,
$$\Delta p_i(v,w)<(\varepsilon-\varepsilon')-(C_\phi+1)(\varepsilon-\varepsilon')=-C_\phi(\varepsilon-\varepsilon')$$
According to this and what we mentioned earlier in the proof,
$$\Delta\Phi_i \leq -C_\phi(\varepsilon-\varepsilon')+\frac{C_3}{c^k-1}$$
and if
$$k\geq \log_c\left(1+\frac{C_3}{(C_\phi-1)(\varepsilon-\varepsilon')}\right)$$
then $\Delta\Phi_i\leq-(\varepsilon-\varepsilon')$, which is a negative constant.
\end{proof}

\textbf{Fixing a violation of Invariant 2.} Next, we consider the second type of maintenance updates, which is to fix the violations of Invariant 2. Whenever a pair $(v,w)$ that violates Invariant 2 is found, the first step is to remove the corresponding edge from its subset $S_i$. Afterwards, we address the same-level violations of Invariant 1 by greedily adding a pair that violates Invariant 1, until none is left. This is the same as performing the edge removal process on the violating pair.

\begin{lemma}%
\label{lem:fix-inv2}%
Let $(v,w)\in S_i$ be an edge that violates Invariant 2, i.e. $d_i^*(v,w)/d(v,w)\leq 1+\varepsilon'$. Also assume that
$$k\geq \log_c\left(1+\frac{2C_5}{\varepsilon-\varepsilon'}\right)$$
Then performing the edge removal process on $(v,w)$ decreases the overall potential $\Phi_i$ of $S_i$ by at least $(\varepsilon-\varepsilon')$.
\end{lemma}
\begin{proof}
Since all the additions and removals in the edge removal process are happening on the same level and also due to the definition of $d_i^*$, there would be no potential change on any of the same or lower level pairs. We know from \cref{lem:single-rem-high} that deleting $(v,w)$ from $S_i$ would increase the potential on higher level pairs by at most $C_5/(c^k-1)$. The potential of the pair itself before the deletion is
$$p_i(v,w)=(1+\varepsilon)-\frac{d_i^*(v,w)}{d(v,w)}$$
After the deletion,
$$p_i(v,w)=C_\phi\cdot\left(\frac{d_i^*(v,w)}{d(v,w)}-(1+\varepsilon')\right)$$
Therefore,
$$\Delta p_i(v,w)=-(\varepsilon-\varepsilon')-(C_\phi+1)\left(1+\varepsilon'-\frac{d_i^*(v,w)}{d(v,w)}\right)$$
We know by the assumption that the stretch of the shortest extended path between $v$ and $w$ would be less than $1+\varepsilon'$, since $(v,w)$ is violating Invariant 2. Therefore,
$$1+\varepsilon'-\frac{d_i^*(v,w)}{d(v,w)}>0$$
Thus,
$$\Delta p_i(v,w)<(\varepsilon-\varepsilon')$$
According to this and what we mentioned earlier in the proof,
$$\Delta\Phi_i \leq -(\varepsilon-\varepsilon')+\frac{C_5}{c^k-1}$$
and if
$$k\geq \log_c\left(1+\frac{2C_5}{\varepsilon-\varepsilon'}\right)$$
then $\Delta\Phi_i\leq-(\varepsilon-\varepsilon')/2$, which is a negative constant.
\end{proof}

\textbf{Bounding the number of updates.} Now that we introduced our maintenance updates and we analyzed the change in the potential functions after each of these updates, we can finally prove our amortized bounds. We prove that the amortized number of edge updates in our algorithm after a point insertion is $\mathcal{O}(1)$, while the amortized number of edge updates after a point deletion is $\mathcal{O}(\log \Delta)$.

\begin{theorem}%
\label{thm:main}%
Our fully-dynamic spanner construction in $d$-dimensional Euclidean spaces has a stretch-factor of $1+\varepsilon$ and a lightness that is bounded by a constant. Furthermore, this construction performs an amortized $\mathcal{O}(1)$ edge updates following a point insertion, and an amortized $\mathcal{O}(\log\Delta)$ edge updates following a point deletion.
\end{theorem}
\begin{proof}
The stretch factor and the lightness immediately follow from the fact that our spanner always satisfies the two invariants, and according to \cref{lem:inv1} and the leapfrog property, that would be enough for a $1+\varepsilon$ stretch factor and constant lightness.

In order to prove the amortized bounds on the number of edge updates after each operation, we recall that by \cref{lem:light-ins-potential}, the potential change $\Delta\Phi^*$ after a point insertion is bounded by a constant, and by \cref{lem:light-del-potential}, the potential change after a point deletion is bounded by $\mathcal{O}(\log \Delta)$. On the other hand, by \cref{lem:fix-inv1} and \cref{lem:fix-inv2}, each maintenance update reduces the potential $\Phi^*$ by at least $(\varepsilon-\varepsilon')/2$, since the impacted $\Phi_i$ reduces after the maintenance update, $\Phi_j$ for $j\neq i$ will remain unchanged, and the extra term $\frac{p_{max}}{2}\cdot\sum_{i=1}^n(D_{max}-\deg_{\calS_1}(v_i))$ will also remain unchanged since the sparse spanner is not affected by the maintenance updates. Therefore, the amortized number of maintenance updates required after each point insertion is $\mathcal{O}(1)$ while this number after a point deletion is $\mathcal{O}(\log \Delta)$. Also, the number of edge updates before the maintenance updates would be bounded by the same amortized bounds. Thus, we can finally conclude that the amortized number of edge updates following a point insertion is $\mathcal{O}(1)$, while for a point deletion it is $\mathcal{O}(\log\Delta)$.
\end{proof}

\section{Conclusions}

In this paper, we presented the first fully-dynamic lightweight construction for $(1+\varepsilon)$-spanners in the $d$-dimensional Euclidean space. In our construction, the amortized number of edge updates following a points insertion is bounded by a constant, and the amortized number of edge updates following a point deletion is bounded by $\mathcal{O}(\log\Delta)$. To achieve this, we defined a set of maintenance updates that could reduce the weight of an existing spanner iteratively, leading to a bounded lightness spanner. We also defined a potential function that could be used to provide an amortized bound on the number of such updates. This framework can be used to find lightweight Euclidean spanners under circumstances other than the fully-dynamic setting, e.g. semi-dynamic or online setting with recourse. It would be interesting to explore other applications of this framework. Since our construction, like the celebrated greedy spanner construction, takes advantage of shortest path queries, it is not necessarily efficient in terms of the running time, and it is suitable when the problem-dependent update cost for a single edge is high. Although we did not focus on optimizing the running time, it would be interesting to look at our construction from that perspective, and find ways to improve its efficiency. We also leave as an open problem whether the amortized bound on the number of edge updates following a point deletion can be improved to $\mathcal{O}(1)$.

\bibliographystyle{plainurl}
\bibliography{main}

\end{document}